\documentclass[11pt]{article}

\usepackage{tikz}
\usepackage{graphicx}
\usepackage{verbatim}
\usepackage{amssymb,amsfonts,amsmath,amsthm}
\usepackage{url}
\usepackage{fullpage}
\usepackage[citecolor=red]{hyperref}
\usepackage[capitalise]{cleveref}

\newcommand{\ip}[1]{\langle #1 \rangle}
\newcommand{\eqdef}{\stackrel{\rm{def}}{=}}

\newcommand{\floor}[1]{{\lfloor#1\rfloor}}

\newcommand{\dist}{{\sf dist}}
\newcommand{\tvdist}{{\sf dist_{TV}}}
\newcommand{\weight}{{\sf weight}}
\newcommand{\A}{{\mathcal A}}
\newcommand{\C}{{\mathcal C}}

\newcommand{\N}{{\mathbb N}}
\newcommand{\R}{{\mathbb R}}

\newcommand{\eps}{\epsilon}
\newcommand{\seq}{\subseteq}

\newcommand{\half}{\frac{1}{2}}

\renewcommand{\int}{{\sf int}}

\newcommand{\D}{{\mathcal D}}

\newcommand{\nmExt}{{\sf nmExt}}
\newcommand{\Ext}{{\sf Ext}}

\newcommand{\supp}{{\sf supp}}

\newcommand{\zo}{\{0,1\}}
\newcommand{\cc}[1][n]{\{0,1\}^{#1}}

\newtheorem{ftheorem}{Theorem}
\newtheorem{flemma}[ftheorem]{Lemma}
\newtheorem{theorem}{Theorem}[section]

\newtheorem{proposition}[theorem]{Proposition}
\newtheorem{definition}[theorem]{Definition}
\newtheorem{lemma}[theorem]{Lemma}
\newtheorem{claim}[theorem]{Claim}

\newenvironment{proof-sketch}{\noindent{\bf Sketch of Proof}\hspace*{1em}}{\qed\bigskip}
\newenvironment{proof-idea}{\noindent{\bf Proof Idea}\hspace*{1em}}{\qed\bigskip}
\newenvironment{proof-of-lemma}[1]{\noindent{\bf Proof of Lemma #1}\hspace*{1em}}{\qed\bigskip}
\newenvironment{proof-of-claim}[1]{\noindent{\bf Proof of Claim #1}\hspace*{1em}}{\qed\bigskip}
\newenvironment{proof-of-thm}[1]{\noindent{\bf Proof of Theorem #1}\hspace*{1em}}{\qed\bigskip}
\newenvironment{proof-attempt}{\noindent{\bf Proof Attempt.}\hspace*{1em}}{\qed\bigskip}

\newenvironment{remark}{\noindent{\it Remark.}}{\bigskip}

\renewcommand{\epsilon}{\varepsilon}

\newcommand*\samethanks[1][\value{footnote}]{\footnotemark[#1]}

\title{An Entropy Lower Bound for Non-Malleable Extractors}

\author{
    Tom Gur
\thanks{{
Department of Electrical Engineering and Computer Sciences, UC Berkeley.\newline
Emails:~\texttt{tom.gur@berkeley.edu,~igors@berkeley.edu}\newline
This work was supported in part by the UC Berkeley Center for Long-Term Cybersecurity.}
}
\and
	Igor Shinkar
\samethanks[1]
}

\begin{document}

\maketitle

\begin{abstract}
A $(k,\eps)$-non-malleable extractor
is a function $\nmExt : \{0,1\}^n \times \{0,1\}^d \to \{0,1\}$
that takes two inputs, a weak source $X \sim \{0,1\}^n$ of min-entropy $k$
and an independent uniform seed $s \in \{0,1\}^d$, and
outputs a bit $\nmExt(X, s)$ that is $\eps$-close to uniform,
even given the seed $s$ and the value $\nmExt(X, s')$
for an adversarially chosen seed $s' \neq s$.
Dodis and Wichs~(STOC 2009) showed the existence of
$(k, \eps)$-non-malleable extractors with
seed length $d = \log(n-k-1) + 2\log(1/\eps) + 6$
that support sources of entropy $k > \log(d) + 2 \log(1/\eps) + 8$.

We show that the foregoing bound is essentially tight,
by proving that any $(k,\eps)$-non-malleable extractor
must satisfy the entropy bound
$k > \log(d) + 2 \log(1/\eps) - \log\log(1/\eps) - C$ 
for an absolute constant $C$.
In particular, this implies that non-malleable extractors
require min-entropy at least $\Omega(\log\log(n))$.
This is in stark contrast to the existence of strong seeded extractors
that support sources of entropy $k = O(\log(1/\eps))$.

Our techniques strongly rely on coding theory.
In particular, we reveal an inherent connection between non-malleable extractors
and error correcting codes, by proving a new lemma which shows that
any $(k,\eps)$-non-malleable extractor with seed length $d$
induces a code $\C \seq \{0,1\}^{2^k}$
with relative distance $0.5 - 2\eps$ and rate $\frac{d-1}{2^k}$.
\end{abstract}

\section{Introduction}\label{sec:intro}

Randomness extractors are central objects in the theory of computation.
Loosely speaking, a \emph{seeded extractor}~\cite{NZ96} is a randomized algorithm
that extracts nearly uniform bits from biased random sources,
using a short seed of randomness.
A \emph{non-malleable extractor}~\cite{DW09} is a seeded extractor
that satisfies a very strong requirement
regarding the lack of correlations of the output of the extractor
with respect to different seeds.

More accurately, a $(k,\eps)$-non-malleable extractor is a
function $\nmExt : \{0,1\}^n \times \{0,1\}^d \to \{0,1\}$
such that for every (weak) source $X$ of min-entropy $k$
and a random variable $s$ uniformly distributed on $\{0,1\}^d$
it holds that $\nmExt(X,s)$ is $\eps$-close to uniform,
even given the seed $s \in \{0,1\}^d$ and the value $\nmExt(X, s')$
for any seed $s' \neq s$ that is determined as an arbitrary function of $s$.
More generally, if $\nmExt(X,s)$ is $\eps$-close to uniform, even given
$\nmExt(X, s'_1), \ldots, \nmExt(X, s'_t)$ for
$t$ adversarially chosen seeds such that $s'_i \neq s$ for all $i\in[t]$,
we say it is a $(k,\eps)$-$t$-non-malleable extractor \cite{CRS14}.

The notion of non-malleable extractors is strongly motivated by
applications to privacy amplification protocols,
as well as proven to be a fundamental notion in the theory of pseudorandomness,
as has been recently exemplified by the key role it played
in the breakthrough construction of explicit two-source extractors
by Chattopadhyay and Zuckerman~\cite{CZ16}.
Moreover, it also has an important connection to Ramsey theory \cite{BKSSW05}.

Non-malleable extractors can be thought of as a strengthening of the notion of
\emph{strong seeded extractors}.
These are functions
$\Ext : \{0,1\}^n \times \{0,1\}^d \to \{0,1\}$
such that for a weak source $X$ and seed $s$
it holds that $\Ext(X,s)$ is $\eps$-close to uniform,
\emph{even given the seed $s \in \{0,1\}^d$}.
We stress that this is a much weaker guarantee than that of non-malleable extractors.
In particular, there exist a blackbox transformation of seeded extractors
into strong seeded extractors with roughly the same parameters \cite{RSW06},
whereas no such transformation is known for non-malleable extractors.

By a simple probabilistic argument (see, e.g.,  \cite{Vadhan12}), there exists a
(strong) seeded extractor $\Ext : \{0,1\}^n \times \{0,1\}^d \to \{0,1\}$
for sources of seed length $d = \log(n) + 2\log(1/\eps) + O(1)$
and min-entropy $k = 2 \log(1/\eps) + O(1)$.
Moreover, by a long line of research, starting with the seminal work of
Nisan and Zuckerman~\cite{NZ96}, and culminating with \cite{GUV09,DKSS13,TU12}
we now know of \emph{explicit} constructions that nearly achieve the optimal parameters.

For non-malleable extractors
the parameters achievable by current constructions are weaker.
Dodis and Wichs showed the existence of $(k, \eps)$-non-malleable extractors
with seed length $d = \log(n-k-1) + 2\log(1/\eps) + 6$,
and entropy $k > \log(d) + 2 \log(1/\eps) + 8$;
and in particular, for $k \geq \log\log(n) + 2 \log(1/\eps)$.
The best explicit construction, due to  \cite{Coh17}
achieve seed length $d = O(\log n) + \tilde{O}(\log(1/\eps))$
for entropy $k = \Omega(d)$.

Note that while for (strong) seeded extractors
there are constructions that support sources of entropy $k = 2 \log(1/\eps) + O(1)$,
without any dependence on $n$,
all known constructions of non-malleable extractors
require the entropy of the source to be at least doubly-logarithmic in $n$.
This naturally raises the question of
whether the dependence on $n$ is indeed necessary for non-malleable extractors.

\begin{center}
    \fbox{\begin{minipage}{0.95\textwidth}
	Question: Is it true that in any $(k, \eps)$-non-malleable extractor
    the entropy $k$ must grow with~$n$?
	\end{minipage}}
\end{center}

In this paper we give a positive answer to this question,
as well as reveal a simple yet fundamental connection
between non-malleable extractors and error-correcting codes,
which we believe to be of independent interest.

\subsection{Our results}
Our main result is a lower bound on the entropy required by non-malleable extractors,
which essentially matches the one obtained by the probabilistic construction.
In particular, we show that any $(k, \eps)$-non-malleable extractor requires
the source entropy $k$ to be at least $\log\log(n) - (2-o_\eps(1)) \log(1/\eps)$.
In fact, we prove the entropy lower bound for the more general notion
of $t$-non-malleable extractors.

\begin{ftheorem}[Main result]\label{thm:t-nmExt lower bound}
    Let $n,k,d,t \in \N$ be parameters such that $t \leq 2^{d/2}$,
    and let $\eps \in (0,c_0)$ for some absolute constant $c_0$.
	If $\nmExt : \{0,1\}^n \times \{0,1\}^d \to \{0,1\}$
	is a $(k,\eps)$-$t$-non-malleable extractor, then
    $d > \log(n-k) + 2\log(1/\eps) - C$ and
	$k \geq \log(d) + 2 \log(1/\eps) - \log\log(1/\eps) + \log(t) - C$
	for an absolute constant $C$.
\end{ftheorem}

We remark that by a recent result of Ben-Aroya et~al.~\cite{BCDLS}
(see \cref{thm:BCDLS}),
the lower bound on $d$ in the theorem is tight up to
an additive factor of $O(\log(t))$, and
our lower bound on $k$ is almost tight in $\eps$,
up to an additive factor of $\log\log(1/\eps)$.
Furthermore, since as we mentioned above, there exist
(strong) seeded extractors for sources of entropy $k = 2 \log(1/\eps) + O(1)$,
\cref{thm:t-nmExt lower bound} implies a chasm between
non-malleable extractors and (strong) seeded extractors;
in particular, it rules out the possibility of transforming seeded extractors
into non-malleable extractors, while preserving the parameters.

A key technical tool that we use to prove \cref{thm:t-nmExt lower bound}
is a lemma, which shows that
any non-malleable extractor induces an error correcting code with a good distance.
We believe this lemma is of independent interest.

\begin{flemma}\label{lemma:nmExt-to-code-intro}
	If there exists a $(k,\eps)$-non-malleable extractor
	$\nmExt : \{0,1\}^n \times \{0,1\}^d \to \{0,1\}$,
    then there exists an error correcting code $\C \seq \{0,1\}^{2^k}$
    with relative distance $0.5 - 2\eps$ and rate $\frac{d-1}{2^k}$.
\end{flemma}

In fact, we actually prove a more general lemma, which shows that
$t$-non-malleable extractors induce codes with rate that grows with $t$.
See \cref{sec:proof} for details.

\subsection{Technical overview}\label{sec:tech}
We provide a high-level overview of the proof of our main result,
the entropy lower bound in \cref{thm:t-nmExt lower bound},
for the simple case of $t=1$ (i.e., for standard non-malleable extractors).
See \cref{sec:proof} for the complete details of the proof
for the general case.
We assume basic familiarity with coding theory and extractors
(see \cref{sec:prelims} for the necessary preliminaries).

Consider a non-malleable extractor $\nmExt$.
Our strategy for showing a lower bound on the source entropy of $\nmExt$
consists of the following two steps.
\begin{enumerate}
	\item Derive a binary code $\C$ with high distance and rate from $\nmExt$, as captured by \cref{lemma:nmExt-to-code-intro}.
	\item Show refined bounds on the rate of binary codes with a given minimum distance, and apply them to $\C$ to obtain an entropy lower bound.
\end{enumerate}
That is, we show that if the parameters of $\nmExt$ were too good,
then the implied code $\C$ would have parameters
that would violate the rate bounds in the second step.
Below, we elaborate on each of the steps.

\paragraph{Deriving codes from non-malleable extractors.}
We start with a $(k,\eps)$-non-malleable extractor
$\nmExt : \{0,1\}^n \times \{0,1\}^d \to \{0,1\}$.
Denote $K = 2^k$, and consider a (flat) source $X$,
which we view as a collection of $K$ vectors $X \seq \{0,1\}^n$.
We show that there is a large subset $S$ of the seeds
such that the evaluations of $\nmExt$, with respect to $X$ and $S$,
constitute a code with high distance and rate.

More accurately, denote by $w^{(s)}$ the \emph{evaluation vector} of $\nmExt$
on the source $X$ and seed $s \in \{0,1\}^d$;
that is, $w^{(s)} = (\nmExt(x,s))_{x \in X}$.
We show that there exists a large subset of seeds $S \subseteq \{0,1\}^d$ such that
\begin{equation*}
	\C \eqdef \{w^{(s)} : s \in S\}
\end{equation*}
is a code with distance $0.5-2\eps$ and rate $(d-1)/K$.

As a warmup, it is instructive to note that
the definition of (standard) \emph{seeded extractors} only requires that
a random coordinate of a random $w^{(s)}$ is nearly uniformly distributed.
\emph{Strong seeded extractors} also imply that most evaluation vectors are roughly
balanced (i.e., contain a similar number of zeros and ones),\footnote{
We stress that elements of a set of nearly-balanced vectors
are not necessarily pairwise-far, unless this set is a \emph{linear space}.
Hence, the foregoing property of strong seeded extractors
does \emph{not} imply a good code in general.}
as a strong seeded extractor needs to output a nearly uniform bit,
even given the seed (i.e., even when the identity of $w^{(s)}$ is known).

The key observation is that the structure of \emph{non-malleable extractors}
asserts that there exists a large subset of seeds
whose corresponding evaluation vectors
are (close to) \emph{pairwise uncorrelated}, and hence
constitute a code with large distance.
Details follow.

Denote the number of seeds by $D = 2^d$.
We wish to show that there exists a subset $S \subset \{0,1\}^d$ of $D/2$ seeds
whose corresponding evaluation vectors are pairwise $(0.5 - 2\eps)$-far.
Suppose the contrary, i.e.,
that \emph{every} set $S$ of $D/2$ seeds contains at least two distinct seeds $s,s'$
such that $w^{(s)}$ is $(0.5 - 2\eps)$-close to $w^{(s')}$.
This means that we can iteratively select a set of $D/2$  ``bad'' seeds
$B \eqdef \{ s_1,\ldots,s_{D/4},s'_1,\ldots,s'_{D/4} \}$ such that
$w^{(s_i)}$ and $w^{(s'_i)}$ are
$(0.5 - 2\eps)$-close in Hamming distance, for every $i\in[D/4]$.
(See \cref{fig:ext_code}.)

\begin{figure}
	\centering
	\includegraphics[scale=0.4]{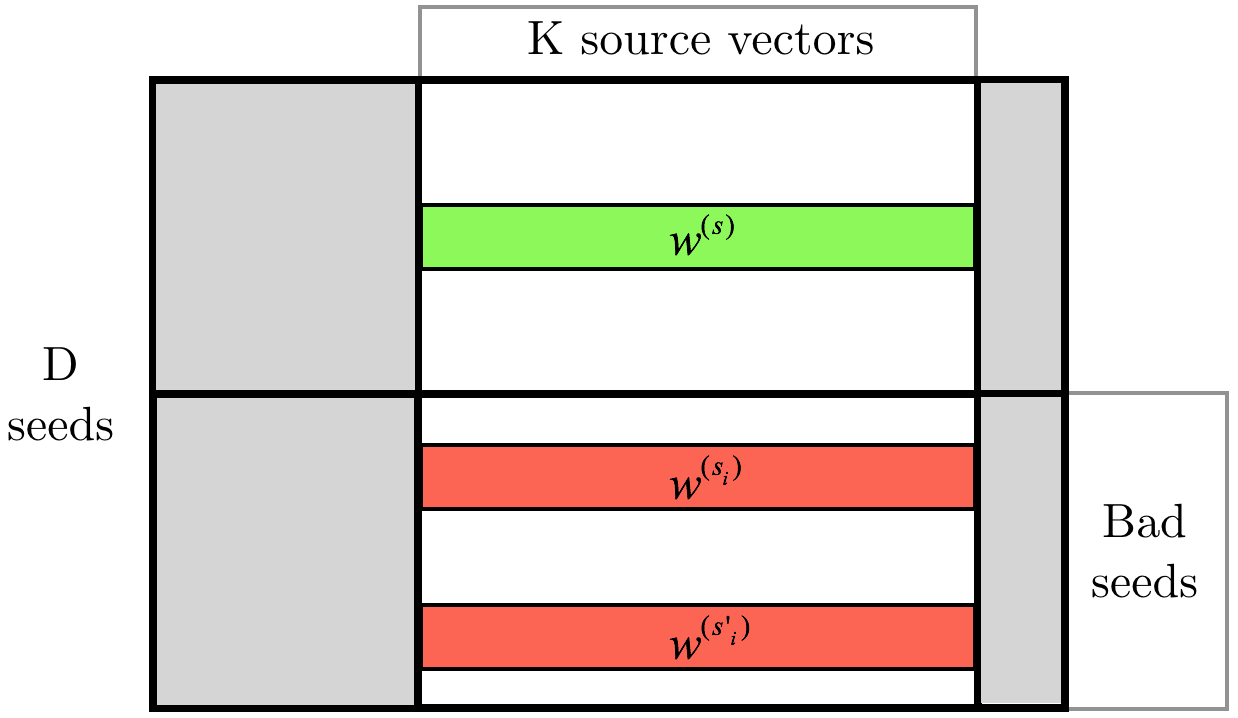}
	\caption{Truth table of a $(k,\eps)$-non-malleable extractor	
	$\nmExt : \{0,1\}^n \times \{0,1\}^d \to \{0,1\}$.
	Rows correspond to the $D = 2^d$ seeds.
	Columns correspond to all $n$-bit vectors,
	out of which we highlight the $K = 2^k$ vectors of the flat source $X$.
	Each vector $w^{(s)} = (\nmExt(x,s))_{x \in X}$ consists of the values
	corresponding to seed $s$ and all vectors of $X$.
	The vectors $w^{(s_i)}$ and $w^{(s'_i)}$
	correspond to a pair of ``bad'' seeds $s_i,s'_i \in B$,
	and hence they are close to each other.}
    \label{fig:ext_code}
\end{figure}

The crux is that having many pairs of correlated evaluation vectors
violates the assumption that $\nmExt$ is a non-malleable extractor.
Intuitively, this holds because for each $w^{(s_i)}$
corresponding to a bad seed $s_i \in B$,
the output of $\nmExt(X,s_i)$ is biased given $\nmExt(X,s'_i)$.
Hence, a non-malleable extractor cannot have a large set of bad seeds.

In \cref{sec:proof-specific} we make this intuition precise
by exhibiting an adversarial function $\A \colon \{0,1\}^d \to \{0,1\}^d$
(with no fixed points) that matches pairs of bad seeds such that
we can construct a distinguisher that,
for a random variable $U_d$ uniformly distributed on the seeds $\{0,1\}^d$,
can tell apart with confidence $\eps$
between $\nmExt(X,U_d)$ and a uniform bit, even when given $\nmExt(X,\A(U_d))$ and~$U_d$.

\paragraph{Refined rate bounds for binary codes.}
After we derived a binary code $\C$ with distance $0.5-2\eps$ and rate $(d-1)/K$
from a $(k,\eps)$-non-malleable extractor $\nmExt$,
we wish to apply upper bounds on the rate of binary codes,
which will in turn imply entropy lower bounds on the entropy that $\nmExt$ requires.

Our starting point is the state-of-the-art upper bound
of McEliece, Rodemich, Rumsey and Welch~\cite{MRRW},
which, loosely speaking, states that
any binary code with relative distance $0.5-\eps$
has rate $O(\eps^2 \log(1/\eps))$ for all sufficiently small $\eps>0$.

Alas, the aforementioned bound does not suffice for the entropy lower bound,
as we need a quantitative bound in terms of the \emph{blocklength} of the code.
We, thus, prove the following theorem, which provides the refined bound that we need.

\begin{ftheorem}\label{thm:MRRW-quant-intro}
    Fix a constant $c \in (0,1/20)$, and let $\eps \in (0,c)$.
    For $K > \frac{c}{\eps^2}$
    let $\C \seq \{0,1\}^K$ be a code with relative distance $\delta = 0.5 - \eps$.
    Then $|\C| < 2^{\frac{23}{c} \eps^2 \log(1/\eps) K}$.
\end{ftheorem}
We prove \cref{thm:MRRW-quant-intro} in \cref{sec:rate_bounds},
relying on the spectral approach of Navon and Samorodnitsky~\cite{NS09}.

To conclude the proof of the entropy lower bound, we argue that if
the non-malleable extractor $\nmExt$ could support entropy that is smaller
than stated in \cref{thm:t-nmExt lower bound},
then the code $\C$ we derive via \cref{lemma:nmExt-to-code-intro} would have
rate that would violate the lower bound in \cref{thm:MRRW-quant-intro}.

\subsection{Organization}
In \cref{sec:prelims} we present the required preliminaries.
In \cref{sec:rate_bounds} we prove the refined bounds on the rate of binary codes.
Finally, in \cref{sec:proof} we prove our main result, \cref{thm:t-nmExt lower bound},
as well as \cref{lemma:nmExt-to-code-intro},
which captures the connection between non-malleable extractors
and error correcting codes.

\section{Preliminaries}\label{sec:prelims}

We cover the notation and basic definitions used in this paper.

\subsection{Notation}
For $n \in \N$, we denote by $[n]$ the set $\{1, \ldots, n\}$,
and by $U_n$ the random variable that is uniformly distributed over $\{0,1\}^n$.
Throughout, $\log(x)$ is defined as $\log_2(x)$.
The \emph{binary entropy function} $H \colon [0,1] \to [0,1]$
is given by $H(x) = -x\log(x) - (1-x)\log(1-x)$.
We denote by ${\bf 1}_{E}$ the \emph{indicator} of an event $E$.
For a finite set $X$, we denote by $\Pr_{x\in X}[\cdot]$
the probability over an element $x$ that is chosen \emph{uniformly at random} from $X$.

\paragraph{Distance.}
The \emph{relative Hamming distance} (or just \emph{distance}),
over alphabet $\Sigma$, between two vectors $x,y \in \Sigma^n$
is denoted $\dist(x,y) \eqdef \frac{|\{i \in[n] \::\: x_i \neq y_i\}|}{n}$.
If $\dist(x,y) \leq \eps$, we say that $x$ is \emph{$\eps$-close} to $y$,
and otherwise we say that $x$ is \emph{$\eps$-far} from $y$.
Similarly, the \emph{relative distance} of $x \in \Sigma^n$
from a non-empty set $S \subseteq \Sigma^n$
is denoted $\dist(x,S) \eqdef \min_{y \in S} \dist(x,y)$.
If $\dist(x,S)  \leq \eps$, we say that $x$ is \emph{$\eps$-close} to $S$,
and otherwise we say that $x$ is \emph{$\eps$-far} from $S$.

The \emph{total variation distance} between two random variables $X_1,X_2$
over domain $\Omega$ is denoted by
$\tvdist(X_1,X_2) \eqdef \sup_{S \subseteq \Omega} \left |\Pr[X_1 \in S] - \Pr[X_2 \in S] \right|$,
and is equivalent, up to a factor $2$, to their $\ell_1$ distance
$\|X_1 - X_2 \|_1 \eqdef \sum_{\omega \in \Omega} \left|\Pr[X_1 = \omega] -\Pr[X_2 = \omega]\right|$.
We say that $X_1$ is \emph{$\eps$-close} to $X_2$ if $\tvdist(X_1,X_2) \leq \eps$,
and otherwise we say that $X_1$ is \emph{$\eps$-far} from $X_2$.

\begin{remark}
In order to show that $X_1$ is \emph{$\eps$-far} from $X_2$
it suffices to show a randomized {distinguisher} $\D \colon \Omega \to \zo$
such that $\left|\Pr[\D(X_1) = 1] - \Pr[\D(X_2) = 1] \right| > \eps$,
where the probabilities are over the random variables $X_1,X_2$ and the randomness of $\D$.
Note that if such randomized distinguisher exists,
then, by averaging, there is also a \emph{deterministic} distinguisher with the same property.
This, naturally, defines the event $S_\D = \{\omega \in \Omega : \D(\omega) = 1\} \seq \Omega$.
for which we have
$\tvdist(X_1,X_2) = \sup_{S \subseteq \Omega} \left |\Pr[X_1 \in S] - \Pr[X_2 \in S] \right|
\geq \left |\Pr[X_1 \in S_\D] - \Pr[X_2 \in S_\D] \right| > \eps$,
and hence $X_1$ is \emph{$\eps$-far} from $X_2$.
\end{remark}

\subsection{Error correcting codes}
Let $k,n\in\N$, and let $\Sigma$ be a finite alphabet.
An \emph{error correcting code} is a set $\C \seq \Sigma^n$,
and the elements of $\C$ are called its \emph{codewords}.
The parameter $n$ is called the \emph{blocklength} of $\C$,
and $k = \log_{|\Sigma|}(|\C|)$ is the \emph{dimension} of $\C$.
The \emph{relative distance} of a code $\C$ is the minimal relative Hamming distance
between its codewords, and is denoted by
$\delta = \min_{c \neq c' \in \C}\{\dist(c,c')\}$.
The \emph{rate} of the code, measuring the redundancy of the encoding,
is the ratio of its dimension and blocklength, and is denote by $\rho = k/n$.
If the alphabet is binary, i.e., $\Sigma = \{0,1\}$,
we say that $\C$ is a \emph{binary code}.

\subsection{Randomness extractors}
We recall the standard definitions of random sources and several types of extractors,
as well as state known bounds that we will need.

\paragraph{Weak sources.}
For integers $n > k$, an \emph{$(n,k)$-random source} $X$ of min-entropy $k$
is a random variable taking values in $\cc$
such that for every $x\in \cc$ is holds that $\Pr[X = x] \leq 2^{-k}$.
An $(n,k)$-random source $X$ is \emph{flat} if it is uniformly
distributed over some subset $S \seq \cc$ of size $2^k$.

It is well known \cite{CG88} that the distribution of any $(n,k)$-random source
is a convex combination of distributions of flat $(n,k)$-random sources,
and thus it typically suffices to consider flat sources.
We follow the literature, restrict our attention to flat $(n,k)$-random sources,
and refer to them simply as \emph{$(n,k)$-sources}.

\paragraph{Seeded extractors.}
A function $\Ext \colon \cc \times \cc[d] \to \zo$
is a \emph{$(k,\eps)$-seeded extractor} if for \emph{any} $(n,k)$-source $X$,
the distribution of $\Ext(X,U_d)$ is $\eps$-close to $U_1$, i.e.,
$\tvdist(\Ext(X,U_d), U_1) \leq \eps$.
(Recall that $U_m$ denotes the random variable
that is uniformly distributed on $\cc[m]$.)

A function $\Ext \colon \cc \times \cc[d] \to \zo$ is a
\emph{$(k,\eps)$-strong seeded extractor} if for any
$(n,k)$-source $X$ the distribution of $(\Ext(X,U_d),U_d)$
 is $\eps$-close to $U_{d+1}$.
We will need the following lower bound on the source entropy required
by strong seeded extractors, due to Radhakrishnan and Ta-Shma~\cite{RT00}
(see also~\cite{NZ96}).
\begin{theorem}[\cite{RT00}~Theorem 1.9]\label{thm:ext-LB}
	Let $\Ext : \{0,1\}^n \times \{0,1\}^d \to \{0,1\}$ be a
	$(k,\eps)$-strong seeded extractor.
	Then, it holds that
	\begin{center}
		$d > \log(n-k) + 2\log(1/\eps) - c$ and
	    $k \geq 2 \log(1/\eps) - c$,
	\end{center}
     for some absolute constant $c \in \R$.
\end{theorem}

\paragraph{Non-malleable extractors.}
Informally, a \emph{non-malleable extractor} $\nmExt$ is a seeded extractor
that for any source $X$ and seed $s$ outputs a bit $\nmExt(X,s)$
that is nearly uniform even if given the seed $s$ and value $\nmExt(X, s')$
for an adversarially selected seed $s'$.

Formally, we say that a function $\A \colon \cc[d] \to \cc[d]$
is an \emph{adversarial function} if it has no fixed points,
i.e., if $\A(s) \neq s$ for all $s \in \cc[d]$.
Non-malleable extractors are defined as follows.

\begin{definition}
    A function $\nmExt \colon \cc \times \cc[d] \to \zo$
    is a \emph{$(k,\eps)$-non-malleable extractor}
    if for any $(n,k)$-source $X$,
    and for any adversarial function $\A \colon \cc[d] \to \cc[d]$,
    it holds that the distribution of the 3-tuple
    $(\nmExt(X,U_d),\nmExt(X,\A(U_d)),U_d)$
    is $\eps$-close to $(U_1,\nmExt(X,\A(U_d)),U_d)$; that is,
    \[
        \dist_{TV}\Big( \big(\nmExt(X,U_d),\nmExt(X,\A(U_d)),U_d \big) \:,\: \big( U_1,\nmExt(X,\A(U_d)),U_d \big)  \Big) \leq \eps.
    \]
\end{definition}

We will also consider the more general notion of \emph{$t$-non-malleable extractors}, in which it is possible to extract randomness even given \emph{multiple} (namely, $t$) outputs of the extractor with respect to adversarially chosen seeds.

\begin{definition}
    A function $\nmExt \colon \cc \times \cc[d] \to \zo$
    is a \emph{$(k,\eps)$-$t$-non-malleable extractor}
    if for any $(n,k)$-source $X$
    and for any $t$ adversarial functions $\A_1,\dots,\A_t \colon \cc[d] \to \cc[d]$
    it holds that
    \[
        \dist_{TV}\Big( \big(\nmExt(X,U_d),(\nmExt(X,\A_i(U_d)))_{i=1}^t,U_d \big)
        \:,\: \big( U_1,(\nmExt(X,\A_i(U_d)))_{i=1}^t,U_d \big)  \Big) \leq \eps.
    \]
\end{definition}

We conclude this section by stating a recent result,
due to Ben-Aroya et~al.~\cite{BCDLS},
extending a result by Dodis and Wichs \cite{DW09},
which complements our \cref{thm:t-nmExt lower bound} by showing that
the lower bound on the seed length $d$ in the \cref{thm:t-nmExt lower bound}
is tight up to an additive factor of $O(\log(t))$, and
the lower bound on $k$ is almost tight in $\eps$,
up to an additive factor of $\log\log(1/\eps)$.

\begin{theorem}[\cite{BCDLS, DW09}]
\label{thm:BCDLS}
    Let $\eps>0$ be sufficiently small, and let $n,k,d,t \in \N$.
    There exists a $(k,\eps)$-$t$-non-malleable extractor $\nmExt : \{0,1\}^n \times \{0,1\}^d \to \{0,1\}$
    with
    \begin{center}
    	$d \leq \log(n) + 2\log(1/\eps) + 2\log(t) + O(1)$ and
	    $k \leq \log(d) + 2 \log(1/\eps) + t + O(\log(t))$.
    \end{center}
\end{theorem}

\section{Refined coding bounds}\label{sec:rate_bounds}

As we mentioned in the technical overview (\Cref{sec:tech}),
we prove our entropy lower bound for non-malleable extractors
by deriving codes from extractors and bounding the rate of these codes.
To this end, in this section we prove refined bounds
on the rate of binary codes with a given minimum distance.
Our starting point is the seminal result of McEliece, Rodemich, Rumsey and Welch~\cite{MRRW}.

\begin{theorem}[\cite{MRRW}]\label{thm:MRRW orig}
    Any code $\C \seq \{0,1\}^n$ with
    relative distance $\delta \in (0,\half)$ has rate at most
    $H\left(\frac{1}{2} - \sqrt{\delta(1-\delta)} \right) + o(1)$,
	where $o(1)$ is some function that tends to zero as $n$ grows to infinity.
\end{theorem}

Observe that in particular, by plugging in $\delta = 0.5 - \eps$
for sufficiently small $\eps>0$, and letting $n$ be sufficiently large
\Cref{thm:MRRW orig} implies that any family of binary codes
with blocklength $n$ and relative distance $\half-\eps$
has rate $\rho = O(\eps^2 \log(1/\eps))$.

However, the above does not suffice for our needs, as to prove our main result (\Cref{thm:t-nmExt lower bound}) we need a quantitative bound on $n$.
We thus prove the following theorem, which provides the refined bound that we seek.

\begin{theorem}\label{thm:MRRW-quant}
    Fix some constant $c \in (0,1/20)$, and let $\eps \in (0,c)$.
    For $n > \frac{c}{\eps^2}$, let $\C \seq \{0,1\}^n$ be a code with relative distance $\delta = \half - \eps$.
    Then $|\C| < 2^{\frac{23}{c} \eps^2 \log(1/\eps) n}$.
\end{theorem}

\begin{proof}
The proof follows the general approach of Navon and Samorodnitsky~\cite{NS09},
who provide a spectral graph theoretic framework to prove upper bounds on
the rate of binary codes.

We will need the following definition, which generalizes
the notion of a \emph{maximal eigenvalue} to subsets of the hypercube.

\begin{definition}
    Let $A \in \{0,1\}^{2^n \times 2^n}$ be the adjacency matrix
    of the hypercube graph; that is,
    $A_{x,y} = 1$ if and only if $x \in \{0,1\}^n$ and $y \in \{0,1\}^n$ differ in exactly one coordinate.
    Given a set $B \seq \{0,1\}^n$, we define
    \[
        \lambda_B = \max_{\substack{f : \{0,1\}^n \to \R \\  \supp(f) \seq B}}  \frac{\ip{A f,f}}{\ip{f,f}} \enspace.
    \]
\end{definition}

To better understand the definition of $\lambda_B$, it is convenient to consider the
subgraph $H_B$ of the hypercube graph $\cc[n]$ induced by the vertices in $B$, and
observe that $\lambda_B$ is the maximal eigenvalue of the adjacency matrix of $H_B$.
Navon and Samorodnitsky~\cite{NS09} prove the following result.

\begin{proposition}[{\cite[Proposition 1.1 ]{NS09}}]
\label{thm:NS rate bound}
  Let $\C \seq \{0,1\}^n$ be a code with relative distance $\delta > 0$,
  and let $\eps > 0$.
  Suppose that for a subset $B \seq \{0,1\}^n$
  it holds that $\lambda_B \geq (1-2\delta + \eps)n$.
  Then $|\C| \leq |B|/\eps$.
\end{proposition}

The foregoing theorem naturally suggest the following proof strategy:
to upper bound the rate of a binary code $\C$
with relative distance $\delta = 0.5 - \eps$,
it suffcies to exhibit a (small as possible) set $B \seq \{0,1\}^n$
whose corresponding maximal eigenvalue satisfies $\lambda_B \geq 3\eps n$;
note that the smaller $B$ is, the better upper bound we get on the rate of $\C$.

Towards this end, let $r \in [n]$ be a parameter to be chosen later,
and let \[ B = \big\{x \in \{0,1\}^n : |x| \in \{r,r+1\} \big\} \enspace.\]
We lower bound the maximal eigenvalue $\lambda_B$ by showing a particular
function $f$ that is supported on $B$, such that
$\frac{\ip{A f,f}}{\ip{f,f}} \geq 3\eps n$. Specifically,
for some $a,b \in \R$ to be chosen later, we define $f: \{0,1\}^n \to \R$ as
\[
    f(x) =
    \begin{cases}
        a & \text{if } |x| = r \\
        b & \text{if } |x| = r+1 \\
        0 & \text{otherwise} \enspace.
    \end{cases}
\]
Clearly $\supp(f) \seq B$. Observe that
\[
  \frac{\ip{A f,f}}{\ip{f,f}}
  = \frac{ab {n \choose r} \cdot (n-r)}{a^2 {n \choose r} + b^2 {n \choose r+1}}
  = \frac{ab {n \choose r} \cdot (n-r)}{a^2 {n \choose r} + b^2 {n \choose r} \cdot \frac{n-r}{r+1}}
  > \frac{ab \cdot r(n-r)}{a^2 \cdot r + b^2 \cdot (n-r)} \enspace.
\]
By choosing $r$ to be an integer in the interval
$\left[ \frac{9 \eps^2}{c} n, \frac{10 \eps^2}{c} n \right]$
and letting $b = a\sqrt{\frac{r}{n}}$ we get that%
\footnote{Note that by the assumption in the theorem
we have $1 < \frac{\eps^2}{c} n < n$. In particular, the interval
$\left[ \frac{9 \eps^2}{c} n, \frac{10 \eps^2}{c} n \right]$ contains an integer.}
\[
  \frac{\ip{A f,f}}{\ip{f,f}}
  > \frac{a^2 \sqrt{r/n}\cdot r(n-r)}{a^2 r + a^2 \cdot (r/n) \cdot (n-r)}
  = \frac{\sqrt{rn} (n-r)}{2n-r}
  > 3 \eps n \enspace,
\]
where the last inequality uses the assumptions that $\eps < c < 1/20$,
which implies that $r \leq \frac{10 \eps^2}{c} n < \frac{n}{2}$.
Therefore, by applying \Cref{thm:NS rate bound} we get that
\[
    |\C| \leq \frac{|B|}{\eps}
    =  \frac{{n \choose r} + {n \choose r+1}}{\eps}
    \leq {n \choose r} \cdot \frac{n}{r\eps}
    \leq \frac{c}{9\eps^3} {n \choose \frac{10\eps^2}{c} n}
    \leq \frac{c}{9\eps^3} \left(\frac{c e}{10\eps^2}\right)^{\frac{10\eps^2}{c} n}
    < 2^{\frac{23\eps^2 \log(1/\eps)}{c} n}
    \enspace,
\]
which concludes the proof of \cref{thm:MRRW-quant}.
\end{proof}

\section{Proof of \Cref{thm:t-nmExt lower bound}}\label{sec:proof}

In this section we prove \Cref{thm:t-nmExt lower bound}, which we restate here
with slightly more specific parameters than those stated above.

\paragraph{\Cref{thm:t-nmExt lower bound} (restated):}
{\em
    Let $n,k,d,t \in \N$ be parameters such that $t \leq 2^{d/2}$,
    and let $\eps \in (0,c_0/2)$ for $c_0 = \min\{1/2^{c},1/20\}$, where $c > 0$ is the constant from \cref{thm:ext-LB}.
	If $\nmExt : \{0,1\}^n \times \{0,1\}^d \to \{0,1\}$
	is a $(k,\eps)$-$t$-non-malleable extractor, then
\begin{center}
	$d > \log(n-k) + 2\log(1/\eps) - O(1)$ and
	$k \geq \log(d) + 2 \log(1/\eps) - \log\log(1/\eps) + \log(t) - O(1)$.
\end{center}
}
\medskip

We start, in \cref{sec:proof-specific},
with the proof of \Cref{thm:t-nmExt lower bound} for the special case where $t=1$
(i.e., for standard non-malleable extractors).
Then, in \cref{sec:proof-general}, we provide the full proof
for general values of $t$.

\subsection{Proof of \cref{thm:t-nmExt lower bound} for $t = 1$}
\label{sec:proof-specific}
Following the outline provided in \cref{sec:tech},
we start the proof with the following lemma,
showing that any non-malleable extractor
induces an error correcting code with good distance.

\begin{lemma}[\cref{lemma:nmExt-to-code-intro}, restated]
\label{lemma:nmExt-to-code}
	If there exists a $(k,\eps)$-non-malleable extractor $\nmExt : \{0,1\}^n \times \{0,1\}^d \to \{0,1\}$,
    then there exists an error correcting code $\C \seq \{0,1\}^{2^k}$ with relative distance $0.5 - 2\eps$
    and rate $\frac{d-1}{2^k}$.
\end{lemma}

\begin{proof}
    Let $\nmExt : \{0,1\}^n \times \{0,1\}^d \to \{0,1\}$ be a $(k,\eps)$-non-malleable extractor, and let $X$ be an $(n,k)$-source.
    That is, $X \seq \{0,1\}^n$ is a collection of $K = 2^k$ vectors,
    which we denote by $X = \{x_1,\dots,x_K\} \seq \{0,1\}^n$.
	For each seed $s \in \{0,1\}^d$, let $w^{(s)} \in \{0,1\}^K$ be the $K$-bit
	\emph{evaluation vector} defined as
    \[w^{(s)} = \big( \nmExt(x_i,s) \big)_{i \in \{1,\dots,K\}} \enspace.\]
    We claim that the (multi-)set $\{w^{(s)} : s \in \{0,1\}^d \} \seq \{0,1\}^K$
    contains an error correcting code $\C \seq \{0,1\}^{K}$ with relative distance $0.5 - 2\eps$ and rate $\frac{d-1}{K}$.
    \begin{claim}\label{claim:nmExt-to-code}
        There exists a subset $S \seq \{0,1\}^d$ of size $2^{d-1}$
        such that for every two distinct $s,s' \in S$ it holds that $\dist(w^{(s)},w^{(s')}) \geq 0.5 - 2\eps$.
    \end{claim}
    \begin{proof}
        Suppose towards contradiction that for every subset $S' \seq \{0,1\}^d$ of size at least $2^{d-1}$
        there exist distinct seeds $s,s' \in S'$ such that $\dist(w^{(s)},w^{(s')}) < 0.5 - 2\eps$.
        We show below that this contradicts the assumption that $\nmExt$ is a $(k,\eps)$-non-malleable extractor.

        Indeed, by the assumption, we can find $s_1,s'_1 \in \{0,1\}^d$ such that
        $\dist(w^{(s_1)}, w^{(s'_1)}) < 0.5 - 2\eps$.
        Then, we can remove $s_1,s'_1$ from $\{0,1\}^d$, and apply the assumption again,
        to obtain $s_2,s'_2 \in \{0,1\}^d \setminus \{s_1,s'_1\}$ such that
        $\dist(w^{(s_2)}, w^{(s'_2)}) < 0.5 - 2\eps$.
        By iteratively repeating this argument $D/4$ times, where $D = 2^d$,
        we obtain $D/4$ pairs of distinct elements
        $(s_1,s'_1),\dots,(s_{D/4},s'_{D/4})$ such that
        \begin{equation}
        \label{eq:badseeds}
        	\forall j \in [D/4] \quad
        	\dist \left( w^{(s_j)}, w^{(s'_j)} \right) < 0.5 - 2\eps \enspace.
        \end{equation}

        Let $B = \{s_j,s'_j : j\in[D/4]\} \seq \{0,1\}^d$
        denote the set of all such ``bad'' seeds,
        and define an adversarial function $\A \colon \cc[d] \to \cc[d]$
        that matches each pair of bad seeds by
        mapping $\A(s_j) = s'_j$ and $\A(s'_j) = s_j$ for all $j\in[D/4]$,
        and defining $\A(s)$ arbitrarily for all other seeds $s \notin B$.

        Next we prove that $\nmExt$ is not a $(k,\eps)$-non-malleable extractor
        by arguing that the distribution of the random variable
        consisting of the $3$-tuple $(\nmExt(X,U_d),\nmExt(X,\A(U_d)),U_d)$
        is $\eps$-far from $(U_1,\nmExt(X,\A(U_d)),U_d)$, where recall that
        $U_m$ denotes the random variable that is uniformly distributed over $\{0,1\}^m$.
        Indeed, consider the following distinguisher
        $\D \colon \zo \times \zo \times \cc[d] \to \{0,1\}$, defined as
        \begin{equation*}
            \D(b,b',s) = \begin{cases}
                          {\bf 1}_{b = b'}, & \mbox{if } s \in B \\
                          U_1, & \mbox{otherwise}\enspace.
                        \end{cases}
        \end{equation*}
        Clearly $\Pr[\D(U_1,\nmExt(X,\A(U_d)),U_d) = 1] = 0.5$.
        On the other hand, by \cref{eq:badseeds}, for $s$ sampled from $U_d$ we have
        \begin{equation*}
            \Pr[\D(\nmExt(X,s),\nmExt(X,\A(s)), s) = 1] \geq (0.5 + 2\eps)\Pr[s \in B]  + 0.5\Pr[s \notin B] \geq 0.5 + \eps
            \enspace,
        \end{equation*}
        thus contradicting the assumption that $\nmExt$ is a $(k,\eps)$-non-malleable extractor.
        This concludes the proof of \cref{claim:nmExt-to-code}.
    \end{proof}
    Therefore, by \cref{claim:nmExt-to-code}
    there exists a set $\C = \{w^{(s)} : s \in S\} \seq \cc[K]$ of size $2^{d-1}$
    such that for every $x,y \in \C$ it holds that $\dist(x,y) \geq 0.5 - 2\eps$,
    i.e., $\C$ is an error correcting code with relative distance $0.5 - 2\eps$ and rate $\frac{d-1}{2^k}$,
    which completes the proof of \cref{lemma:nmExt-to-code}.
\end{proof}

By applying the bound from \cref{thm:MRRW-quant} to the code obtained in \cref{lemma:nmExt-to-code},
we prove \cref{thm:t-nmExt lower bound} for the case of $t=1$.
\begin{proof}[Proof of \cref{thm:t-nmExt lower bound} for $t=1$]
	Since every non-malleable extractor is, in particular, a strong seeded extractor,
	then by \cref{thm:ext-LB} it holds that the seed length
	is $d > \log(n-k) + 2\log(1/\eps) - c$, as required.
	Furthermore, \cref{thm:ext-LB} also implies that
	\begin{equation}
	\label{eq:ext-LB-k}
	    k \geq 2 \log(1/\eps) - c.
	\end{equation}
	
    By \cref{lemma:nmExt-to-code},
	if $\nmExt : \{0,1\}^n \times \{0,1\}^d \to \{0,1\}$
	is a $(k,\eps)$-non-malleable extractor,
    then there exists an error correcting code $\C \seq \{0,1\}^{2^k}$
    with relative distance $0.5 - 2\eps$ and rate $\frac{d-1}{2^k}$.

    Next, we wish to apply \cref{thm:MRRW-quant} to the code $\C$.
    Recall that by the assumption it holds that $\eps < c_0$ and $c_0 < 1/2^c$,
    and observe that by \cref{eq:ext-LB-k} we have
    $2^k \geq \frac{2^{-c}}{\eps^2} > \frac{c_0}{\eps^2}$.
    Therefore, by applying \cref{thm:MRRW-quant},
    with respect to $c_0$ (recall that $c_0 < 1/20$) and $2\eps < c_0$
    we get that
    \[2^{d-1} \leq |\C| < 2^{\frac{23}{c_0}\cdot (2\eps)^2 \log(1/2\eps) 2^k} \enspace, \]
    and thus $k \geq \log(d) + 2\log(1/\eps) - \log\log(1/\eps) - O(1)$, as required.
\end{proof}

\subsection{Proof of \cref{thm:t-nmExt lower bound} for general $t$}
\label{sec:proof-general}
Next, we extend the idea  presented in \cref{sec:proof-specific} to larger values of $t$.
The key step is the following lemma.

\begin{lemma}\label{lemma:t-nmExt-to-t-code}
	If there exists a $(k,\eps)$-$t$-non-malleable extractor $\nmExt : \{0,1\}^n \times \{0,1\}^d \to \{0,1\}$,
    then, there exists an error correcting code $\C \seq \{0,1\}^{2^k}$ with relative distance $0.5 - 2\eps$
    such that $|\C| \geq (2^{d-1}/t)^{\floor{t/2}}$.
\end{lemma}

\begin{proof}
	Let $\nmExt : \{0,1\}^n \times \{0,1\}^d \to \{0,1\}$
	be a $(k,\eps)$-$t$-non-malleable extractor.
    Similarly to the proof of \cref{lemma:nmExt-to-code},
    we set $K = 2^k$, and let $X$ be an $(n,k)$-source, which we view as
    a collection of vectors $X =\{x_1,\dots,x_K\} \seq \{0,1\}^n$.
	For each seed $s \in \{0,1\}^d$, let $w^{(s)} \in \{0,1\}^K$
	be the $K$-bit \emph{evaluation vector},
	defined as $w^{(s)} = \big(\nmExt(x_i,s)\big)_{i \in \{1,\dots,K\}}$.
	Hereafter, all sums involving binary vectors are summations over $\mathsf{GF}(2)$.
	For $x \in \cc$, we denote by $\weight(x)$ the (absolute) Hamming weight of $x$.
	
	Whereas before, in the proof of \cref{lemma:nmExt-to-code}, we showed that
	the multi-set of evaluation vectors
	$\big\{w^{(s)} : s \in \{0,1\}^d \big\} \seq \{0,1\}^K$
	simply contains an error correcting code with good parameters,
	here we will derive our code by considering all $\mathsf{GF}(2)$-linear
	combinations of $\floor{t/2}$ elements
	of a carefully selected subset of the evaluation vectors.
	
	Towards that end, the next claim shows that there exists a large subset of seeds
	such that any linear combination of $t+1$ of the evaluation vectors
	that corresponds to these seeds has large Hamming weight.
	
    \begin{claim}\label{claim:t-nmExt-to-t-code}
        There is a subset $S \seq \{0,1\}^d$ of size $2^{d-1}$
        such that for every subset $I \seq S$ of size $|I| \leq t+1$ it holds that
        $\weight \left(\sum_{s \in I}w^{(s)} \right) \geq (0.5 - 2\eps)K$.
    \end{claim}

    \begin{proof}
        Assume towards contradiction that for every subset $S' \seq \{0,1\}^d$ of size at least $2^{d-1}$
        there are $t' \leq t+1$ distinct seeds $s_{1}\ldots,s_{t'} \in S'$ such that
        \begin{equation*}
        	\Pr_{x \in X} \left[\sum_{j=1}^{t'}\nmExt\left(x,s_{j}\right) = 0 \right] < 0.5 - 2\eps \enspace.	
        \end{equation*}
        We show below that this contradicts the assumption that $\nmExt$ is a $(k,\eps)$-$t$-non-malleable extractor.

        By our assumption, there is a subset of seeds $S_1\seq \{0,1\}^d$
        for which there exists $I_1\seq S_1$ of size $|I_1| = t'_1 \leq t$
        such that $\weight \left(\sum_{s \in I_1} w^{(s)} \right) < (0.5 - 2\eps)K$.
        We remove $I_1$ from $\{0,1\}^d$, and apply the assumption again
        to obtain $I_2 \seq \{0,1\}^d$ of size $|I_2| = t'_2 \leq t+1$ such that
        $\weight \left(\sum_{s \in I_2} w^{(s)} \right) < (0.5 - 2\eps)K$.
        We then remove $I_2$ from $\cc[d] \setminus I_1$, and apply the assumption again
        with respect to $\cc[d] \setminus (I_1 \cup I_2)$.
        By repeating this argument as long as $| \cup_j I_j| < 2^{d-1}$,
        we obtain $R$ disjoint subsets $I_1,\dots,I_R$,
        where the size of each $I_j$ is $t'_j \leq t+1$,
        such that $\sum_{j=1}^R |I_j| \geq 2^{d-1}$
        and
    	\begin{equation}
    	\label{eq:weight}
			\weight \left(\sum_{s \in I_j} w^{(s)} \right) < (0.5 - 2\eps)K \enspace,
    	\end{equation}
    	for all $j \in [R]$. Analogously to the proof of \cref{lemma:nmExt-to-code}, the set $I_1 \cup \ldots \cup I_T$ consists of the ``bad seeds'' that correspond to evaluation vectors whose $(t+1)$-element linear combinations are of low weight.

		To prove that the foregoing collection of ``bad seeds'' violates the
		assumption that $\nmExt$ is a $(k,\eps)$-$t$-non-malleable extractor,
		we exhibit $t$ adversarial functions $\A_1,\dots,\A_t \colon \cc[d] \to \cc[d]$
		(with no fixed points)
		for which there exists a function that distinguishes between the random variables
        consisting of the $(t+2)$-tuples
        \begin{center}
        	$\left(\nmExt(X,U_d),\Big(\nmExt\big(X,\A_\ell(U_d)\big)\Big)_{\ell \in [t]},U_d\right)$
        	and $\left(U_1,\Big(\nmExt\big(X,\A_\ell(U_d)\big)\Big)_{\ell \in [t]},U_d\right)$
        \end{center}
		with confidence $\eps$, where recall that
        $U_m$ denotes the random variable that is uniformly distributed over $\{0,1\}^m$.	
		
        We define the family $\{\A_\ell\}_{\ell\in[t]}$ in the natural way,
        by mapping each of the bad seeds to the set of seeds
        with which its linear combination is a low weight vector. That is,
        for each $j \in [R]$ let $I_j = \{s_1,\dots,s_{t'_j}\}$, where $t'_j \leq t+1$.
        Then, for all $\ell \in [t]$ we define
        \begin{equation*}
        	\A_\ell(s_i) =
        	\begin{cases}
        		s_{i+\ell \pmod{t'_j}}, & \text{for } s_i \in I_j,\: j\in[R] \\
        		\mbox{arbitrary}, & \text{for } s \in \cc[d] \setminus \left(\cup_{j\in[R]} I_j\right) \enspace.
        	\end{cases}
        \end{equation*}
        Note that by definition of the $\A_\ell$'s,
        for all $j \in [R]$ and $s \in I_j$ it holds that
        $\{s\} \cup \{\A_\ell(s)\}_{\ell \in [t'_j - 1]} = I_j$, and so, by \cref{eq:weight} we have that
        \[
            \Pr_{x \in X} \left[ \nmExt(x,s) = \sum_{i=1}^{t'_j-1}\nmExt\big(x,\A_i(s)\big) \right]
            = \frac{\weight \left(\sum_{s \in I_j} w^{(s)} \right)}{K} < (0.5 - 2\eps)K \enspace.
        \]
        Next, we define the distinguisher
        $\D \colon \zo \times \zo^t \times \cc[d] \to \{0,1\}$ as
        \begin{equation*}
            \D(b,b_1,\dots,b_t,s) = \begin{cases}
                          {\bf 1}_{b = \sum_{i \in [t'_j-1]}b_i}, & \mbox{if $s \in I_j$ for some $j \in [R]$}  \\
                          U_1, & \mbox{otherwise} \enspace.
                        \end{cases}
        \end{equation*}
        Clearly $\Pr\Big[\D\Big(U_1,\big(\nmExt(X,\A_\ell(U_d))\big)_{\ell \in [t]},U_d\Big) = 1\Big] = 0.5$.
        On the other hand, for $s$ sampled from $U_d$ we have
        \begin{flalign*}
            &\Pr\left[\D\Big(\nmExt(X,s),\big(\nmExt(X,\A_\ell(s))\big)_{\ell\in[t]}, s\Big) = 1\right]  \\
            &\geq (0.5 + 2\eps)\Pr \left[s \in \cup_{j \in [R]}I_j \right]
            + 0.5\Pr \left[ s \in \cc[d] \setminus \cup_{j \in [R]}I_j \right] \geq 0.5 + \eps
            \enspace,
        \end{flalign*}
        thus contradicting the assumption that $\nmExt$ is a $(k,\eps)$-$t$-non-malleable extractor.
        This concludes the proof of \cref{claim:t-nmExt-to-t-code}.
    \end{proof}
    Let $S \seq \cc[d]$ be the set guaranteed by \cref{claim:t-nmExt-to-t-code}, and consider the code
    \begin{equation*}
        \C \eqdef \left\{ \sum_{s \in I} w^{(s)} : I \seq S, |I| \leq \floor{t/2} \seq \cc[K]\right\}
        \enspace.
    \end{equation*}
    Note that for $D = 2^d$ we have $|\C| \geq {D/4 \choose \floor{t/2}} \geq (D/2t)^{\floor{t/2}}$.
    By the guarantee of \cref{claim:t-nmExt-to-t-code},
    for every distinct $x,y \in \C$ it holds that
    $\dist(x,y) \geq 0.5 - 2\eps$; that is
    $\C \seq \cc[K]$ is an error correcting code with relative distance $0.5 - 2\eps$,
    which completes the proof of \cref{lemma:t-nmExt-to-t-code}.
\end{proof}

We prove \cref{thm:t-nmExt lower bound} by applying the bound from \cref{thm:MRRW-quant} to the code obtained in \cref{lemma:t-nmExt-to-t-code}, analogously to the way we proved the theorem for the restricted case of $t=1$ before.

\begin{proof}[Proof of \cref{thm:t-nmExt lower bound} (general case)]
Since every $t$-non-malleable extractor is, in particular, a strong seeded extractor,
	then by \cref{thm:ext-LB} it holds that the seed length
	is $d > \log(n-k) + 2\log(1/\eps) - c$, as required.
	Furthermore, \cref{thm:ext-LB} also implies that $k \geq 2 \log(1/\eps) - c$.

    By \cref{lemma:t-nmExt-to-t-code},
	if $\nmExt : \{0,1\}^n \times \{0,1\}^d \to \{0,1\}$
	is a $(k,\eps)$-non-malleable extractor,
    then there exists an error correcting code $\C \seq \{0,1\}^{2^k}$
    with relative distance $0.5 - 2\eps$ such that $|\C| \geq (2^{d-1}/t)^{\floor{t/2}}$.

    We wish to apply \cref{thm:MRRW-quant} to the code $\C$.
    Recall that by the assumption it holds that $\eps < c_0$ and $c_0 < 1/2^c$,
    and observe that according to the bound on $k$ given by \cref{thm:ext-LB},
    we have that $2^k \geq \frac{2^{-c}}{\eps^2} > \frac{c_0}{\eps^2}$.
    Therefore, by applying \cref{thm:MRRW-quant},
    with respect to $c_0$ (recall that $c_0 < 1/20$) and $2\eps < c_0$,
    we get that
    \[(2^{d-1}/t)^{\floor{t/2}} \leq |\C| < 2^{\frac{23}{c_0}\cdot (2\eps)^2 \log(1/2\eps) 2^k} \enspace, \]
	and by the assumption that $\log(t) < d/2$ we get that
	\begin{equation*}
		\frac{23}{c_0} \cdot (2\eps)^2 \log(1/2\eps) 2^k \geq \big(d-2-\log(t)\big) \cdot \floor{t/2} \geq \Omega(d \cdot t) \enspace.
	\end{equation*}
    This implies that $k \geq \log(d) + \log(t) + 2\log(1/\eps) - \log\log(1/\eps) - O(1)$, as required.
\end{proof}

\section*{Acknowledgements}
We are thankful to Gil Cohen for helpful discussions.
We also thank Venkatesan Guruswami for a discussion regarding the MRRW bounds.

\bibliographystyle{alpha}
\bibliography{nmExtLB}

\newcommand{\etalchar}[1]{$^{#1}$}
\begin{thebibliography}{MRR{\etalchar{+}}77}

\bibitem[BCD{\etalchar{+}}17]{BCDLS}
Avraham {Ben-Aroya}, Eshan Chattopadhyay, Dean Doron, Xin Li, and Amnon
  {Ta-Shma}.
\newblock A reduction from efficient non-malleable extractors to low-error
  two-source extractors with arbitrary constant rate.
\newblock {\em ECCC TR17-027}, 2017.
\newblock Manuscript.

\bibitem[BKS{\etalchar{+}}05]{BKSSW05}
Boaz Barak, Guy Kindler, Ronen Shaltiel, Benny Sudakov, and Avi Wigderson.
\newblock Simulating independence: New constructions of condensers, ramsey
  graphs, dispersers, and extractors.
\newblock In {\em Proceedings of the thirty-seventh annual ACM symposium on
  Theory of computing}, pages 1--10, 2005.

\bibitem[CG88]{CG88}
Benny Chor and Oded Goldreich.
\newblock Unbiased bits from sources of weak randomness and probabilistic
  communication complexity.
\newblock {\em SIAM Journal on Computing}, 17(2):230--261, 1988.

\bibitem[Coh17]{Coh17}
Gil Cohen.
\newblock Towards optimal two-source extractors and ramsey graphs.
\newblock In {\em Proceedings of the 49th Annual ACM SIGACT Symposium on Theory
  of Computing}, pages 1157--1170, 2017.

\bibitem[CRS14]{CRS14}
Gil Cohen, Ran Raz, and Gil Segev.
\newblock Nonmalleable extractors with short seeds and applications to privacy
  amplification.
\newblock {\em SIAM Journal on Computing}, 43(2):450--476, 2014.

\bibitem[CZ16]{CZ16}
Eshan Chattopadhyay and David Zuckerman.
\newblock Explicit two-source extractors and resilient functions.
\newblock In {\em Proceedings of the 48th Annual ACM SIGACT Symposium on Theory
  of Computing}, pages 670--683, 2016.

\bibitem[DKSS13]{DKSS13}
Zeev Dvir, Swastik Kopparty, Shubhangi Saraf, and Madhu Sudan.
\newblock Extensions to the method of multiplicities, with applications to
  kakeya sets and mergers.
\newblock {\em SIAM Journal on Computing}, 42(6):2305--2328, 2013.

\bibitem[DW09]{DW09}
Yevgeniy Dodis and Daniel Wichs.
\newblock Non-malleable extractors and symmetric key cryptography from weak
  secrets.
\newblock In {\em Proceedings of the Forty-first Annual ACM Symposium on Theory
  of Computing}, STOC '09, pages 601--610, New York, NY, USA, 2009. ACM.

\bibitem[GUV09]{GUV09}
Venkatesan Guruswami, Christopher Umans, and Salil Vadhan.
\newblock Unbalanced expanders and randomness extractors from parvaresh--vardy
  codes.
\newblock {\em Journal of the ACM (JACM)}, 56(4):20, 2009.

\bibitem[MRR{\etalchar{+}}77]{MRRW}
Robert~J. Mceliece, Eugene~R. Rodemich, Howard Rumsey, Lloyd, and R.~Welch.
\newblock New upper bounds on the rate of a code via the {Delsarte-MacWilliams}
  inequalities.
\newblock {\em IEEE Transactions on Information Theory}, IT-23(2):157--166,
  1977.

\bibitem[NS09]{NS09}
Michael Navon and Alex Samorodnitsky.
\newblock Linear programming bounds for codes via a covering argument.
\newblock {\em Discrete {\&} Computational Geometry}, 41(2):199--207, 2009.

\bibitem[NZ96]{NZ96}
Noam Nisan and David Zuckerman.
\newblock Randomness is linear in space.
\newblock {\em Journal of Computer and System Sciences}, 52(1):43--52, 1996.

\bibitem[RSW06]{RSW06}
Omer Reingold, Ronen Shaltiel, and Avi Wigderson.
\newblock Extracting randomness via repeated condensing.
\newblock {\em SIAM Journal on Computing}, 35(5):1185--1209, 2006.

\bibitem[RT00]{RT00}
Jaikumar Radhakrishnan and Amnon {Ta-Shma}.
\newblock Bounds for dispersers extractors and depth-two super concentrators.
\newblock {\em SIAM J. Discrete Math.}, 13(1):2--24, 2000.

\bibitem[TU12]{TU12}
Amnon {Ta-Shma} and Christopher Umans.
\newblock Better condensers and new extractors from parvaresh-vardy codes.
\newblock In {\em Computational Complexity (CCC), 2012 IEEE 27th Annual
  Conference on}, pages 309--315, 2012.

\bibitem[Vad12]{Vadhan12}
Salil Vadhan.
\newblock Pseudorandomness.
\newblock {\em Foundations and Trends in Theoretical Computer Science},
  7(1--3):1--336, 2012.

\end{thebibliography}

\end{document}